\UseRawInputEncoding
\documentclass[notitlepage,aps,pra,twocolumn,groupaddress,superscriptaddress,10pt]{revtex4-1}
\pdfoutput=1
\usepackage{stmaryrd}
\usepackage{amssymb,amsmath,amsthm,amsfonts,amsbsy}
\usepackage{bm,bibunits,color,chngcntr,epsfig,epstopdf,graphicx,dsfont,float}
\usepackage{hyperref,marvosym,lipsum,makecell,mathrsfs,rotating}
\usepackage[english]{babel}
\usepackage[normalem]{ulem}

\newcommand{\be}{\begin{equation}}
\newcommand{\ee}{\end{equation}}
\newcommand{\bea}{\begin{eqnarray}}
\newcommand{\eea}{\end{eqnarray}}
\newcommand{\ket}{\rangle}
\newcommand{\bra}{\langle}

\newcommand{\I}{\mathds{1}}
\newcommand{\ra}{\rightarrow}

\def\C#1{\mathcal #1}

\definecolor{gray}{gray}{0.9}

\usepackage{setspace}

\begin{document}
\newtheorem{theorem}{Theorem}
\newtheorem{prop}[theorem]{Proposition}
\newtheorem{corollary}[theorem]{Corollary}
\newtheorem{open problem}[theorem]{Open Problem}
\newtheorem{conjecture}[theorem]{Conjecture}
\newtheorem{definition}{Definition}
\newtheorem{remark}{Remark}
\newtheorem{example}{Example}
\newtheorem{task}{Task}

\title{Quantum resource theory of coding for error correction}

\author{Dong-Sheng Wang}
\email{wds@itp.ac.cn}
\affiliation{CAS Key Laboratory of Theoretical Physics, Institute of Theoretical Physics,
Chinese Academy of Sciences, Beijing 100190, China}
\affiliation{School of Physical Sciences, University of Chinese Academy of Sciences, Beijing 100049, China}
\author{Yuan-Dong Liu}
\email{liuyuandong@itp.ac.cn}
\affiliation{CAS Key Laboratory of Theoretical Physics, Institute of Theoretical Physics,
Chinese Academy of Sciences, Beijing 100190, China}
\affiliation{School of Physical Sciences, University of Chinese Academy of Sciences, Beijing 100049, China}
\author{Yun-Jiang Wang}
\email{yunjiangw@xidian.edu.cn}
\affiliation{School of Telecommunication Engineering, Xidian University, Xi’an, Shann Xi 710071, China}
\affiliation{Guangzhou Institute of Technology, Xidian University, Guangzhou 510555, China}
\affiliation{Hangzhou Institute of Technology, Xidian University, Hangzhou, Zhejiang 311231, China}
\author{Shunlong Luo}
\email{luosl@amt.ac.cn}
\affiliation{Academy of Mathematics and Systems Science, Chinese Academy of Sciences, Beijing 100190, China}
\date{\today}
\begin{abstract}
    Error-correction codes are central for fault-tolerant information processing.
    Here we develop a rigorous framework to describe various coding models 
    based on quantum resource theory of superchannels.
    We find, by treating codings as superchannels, 
    a hierarchy of coding models can be established,
    including the entanglement assisted or unassisted settings,
    and their local versions. 
    We show that these coding models can be  
    used to classify error-correction codes and 
    accommodate different computation and communication settings 
    depending on the data type, side channels, and pre-/postprocessing.
    We believe the coding hierarchy could also inspire new coding models and error-correction methods.
\end{abstract}

\maketitle

\begin{spacing}{1.0}

\section{Introduction}

To protect information against noises, 
error-correction codes are needed~\cite{RL09}.
This is important for both reliable computation and communication.
A fundamental quantity of a noise channel is its capacity,
which is the maximal rate of information transmission.
Shannon’s classical channel coding theorem shows that
the capacity is given by the maximal mutual information 
$I(A:B)$ between the sender Alice, $A$, and receiver Bob, $B$,
over all possible input~\cite{Sha48}.
Quantum noises are mathematically modeled as 
quantum channels,
and a variety of capacities can be 
defined~\cite{Hol99,BNS98,BSS+99,Dev05}.

A notable difference between the classical and quantum cases
are the non-additivity of some channel capacities~\cite{Wat18,Wil17},
leading to many interesting phenomena~\cite{SY08,LLS+14}. 
The primary setting of quantum communication is the direct transmission of quantum data,
with the asymptotic coherent information as the measure of its quantum capacity~\cite{BNS98}.
Many efforts have been devoted to understand this. 
For instance, it is shown that 
weakly noisy channels are limited in their non-additivity~\cite{LLS18},
and many non-additive examples are constructed~\cite{LLS23}.
It was found that adding forward classical communication does not make a difference,
but the back communication can increase its quantum capacity~\cite{BDS97}.
Meanwhile, other quantities are also explored, such as the 
reverse coherent information~\cite{GPL09} and Rains information~\cite{TWW17}. 

In recent years, the quantum resource theory (QRT) has become 
a unifying framework to characterize quantum features~\cite{CG19}. 
A notable example is bipartite entanglement,
for which the bipartite separable states are free, i.e., of zero resource,
and stochastic local operations with classical communication (SLOCC)
are free operations that cannot increase entanglement~\cite{HHH+09}.
For quantum computation,
we recently developed the QRT of modeling and logical gates~\cite{W21_model,W23_ur,W23_qvn} 
which can be used to classify universal quantum computing models
and predict new ones.

There have been a few resource-theoretic approaches for quantum communication, 
e.g., Refs.~\cite{DHW04,TWH20,KCS20,LW19},
and they are mostly QRT of channels~\cite{CG19}.
In these approaches, some channels are identified as useless for communication tasks,
such as the set of replacement channels~\cite{TWH20},
and unitary channels, with maximal quantum capacity
for a given dimension, will be the most resourceful. 
However, this cannot identify the codings that achieve 
a quantum capacity of a channel as resources. 
This motivates us to consider the QRT of codings, 
which are actually superchannels~\cite{CDP08a}.

In this paper, we develop a QRT of codings which are relevant to 
quantum communication and error correction,
and especially can be applied to various situations. 
A coding, including a pair of encoding and decoding operations,
is a superchannel which converts a channel into another.
A free set is the codings that do not work well for
a given noise channel, 
and resources are those that indeed work.
The one-side computational ability of Alice and Bob also matters,
and different codes can be chosen for different purposes~\cite{RL09}.
We find this can be characterized by a hierarchical family of coding models.


We develop a hierarchy of coding models that includes 
the standard settings for quantum communication
and entanglement-assisted one~\cite{Wat18}, and also our two recent refined models~\cite{W22_ca}.
The hierarchy is defined according to two notions of locality,
a one-side computational locality, relevant to Alice's and Bob's computational ability,
and a communication locality, relevant to the entanglement assistance. 
The hierarchy has a nontrivial structure.
This can be seen by their capacities. 
A capacity can be understood as a measure of the maximal resource of codings
in a coding model.
The quantum capacities of a channel $\Phi$ 
are $I(\Phi)$, $(\log d+ I(\Phi))/2$, $I_c(\Phi)$, and $J(\Phi)/2$, respectively,
for models from I to IV,
and for $I(\Phi)$ as the coherent information with maximally mixed state as input,
$I_c(\Phi)$ denoting the asymptotic coherent information,
$J(\Phi)$ as the quantum mutual information (see Table~\ref{tab:capacity}). 
However, $I_c(\Phi)$ is non-additive 
and there is no definite order between $I_c(\Phi)$ and $(\log d+ I(\Phi))/2$.
Therefore, the models I, II, and IV form a sub-hierarchy,
and the models I, III, and IV also form a sub-hierarchy.
Moreover, there could also be more models in this coding family
based on our study of the gaps among those capacities.

This paper contains the following parts. 
Section~\ref{sec:pre} provides the necessary background,
and section~\ref{sec:coding} introduces our definition of coding models.
The quantum capacities of these four models are studied in section~\ref{sec:capacity},
and the resource theory of them are established in section~\ref{sec:qrt}.
We also carry out numerical simulation to study the gaps among 
these quantum capacities for the qubit case in section~\ref{sec:num}.
We conclude in section~\ref{sec:conc}.

\section{Preliminaries}
\label{sec:pre}

\subsection{Quantum channels and superchannels}

A quantum channel $\Phi$ acting on a finite-dimensional Hilbert space $\C X$
is a completely positive, trace-preserving (CPTP) map of the form 
\be \Phi(\rho)= \sum_i K_i \rho K_i^\dagger, \ee
with $K_i$ known as Kraus operators satisfying $\sum_i K_i^\dagger K_i= \I_d$ (identity operator of dimension $d=\text{dim}\C X$),
and states $\rho \in \C D(\C X)$ as trace-class nonnegative semidefinite operators.
We often ignore the subscript of $\I_d$ if there is no confusion. 
Also we will use $\pi_d$ to denote the completely mixed state of dimension $d$,
and simply as $\pi$ if the dimension is implicit.
The above formalism can also be extended to describe channels 
$\Phi \in C(\C X,\C Y)$ that do not necessarily preserve dimension. 

A different representation of channel $\Phi$ is 
an isometry $V$ with
$\Phi(\rho)=\text{tr}_\text{E} V\rho V^\dagger$ for E denoting an `environment' (or `Eve').
The isometry $V$ is formed by Kraus operators $V=\sum_i |i\ket \otimes K_i$
for $|i\ket$ as orthogonal states of E. 
The isometry $V$ forms a part of a unitary circuit $U$ with $V=U|0\ket$, 
for $|0\ket$ as a state of E.
A complementary channel $\Phi^c$ can be defined as the map from the input to E,
and the state of E is 
\be \rho_\text{E}=\sum_{ij} \text{tr}(\rho K_i^\dagger K_j) |j\ket \bra i|. \ee

The notable channel-state duality~\cite{Jam72,Cho75,BZ06,JLF13} also enables the representation of $\Phi$ as a state
\be \omega_{\Phi} := (\Phi \otimes \I) (\omega), \ee 
usually known as a Choi state, for $\omega:=|\omega\ket \bra \omega|$,
and $|\omega\ket:= \sum_i |ii\ket /\sqrt{d}$ known as the ebit.
The rank of the channel $\text{r}(\Phi)$ is the rank of $\omega_{\Phi}$.

\begin{figure}[t!]
    \centering
    \includegraphics[width=0.3\textwidth]{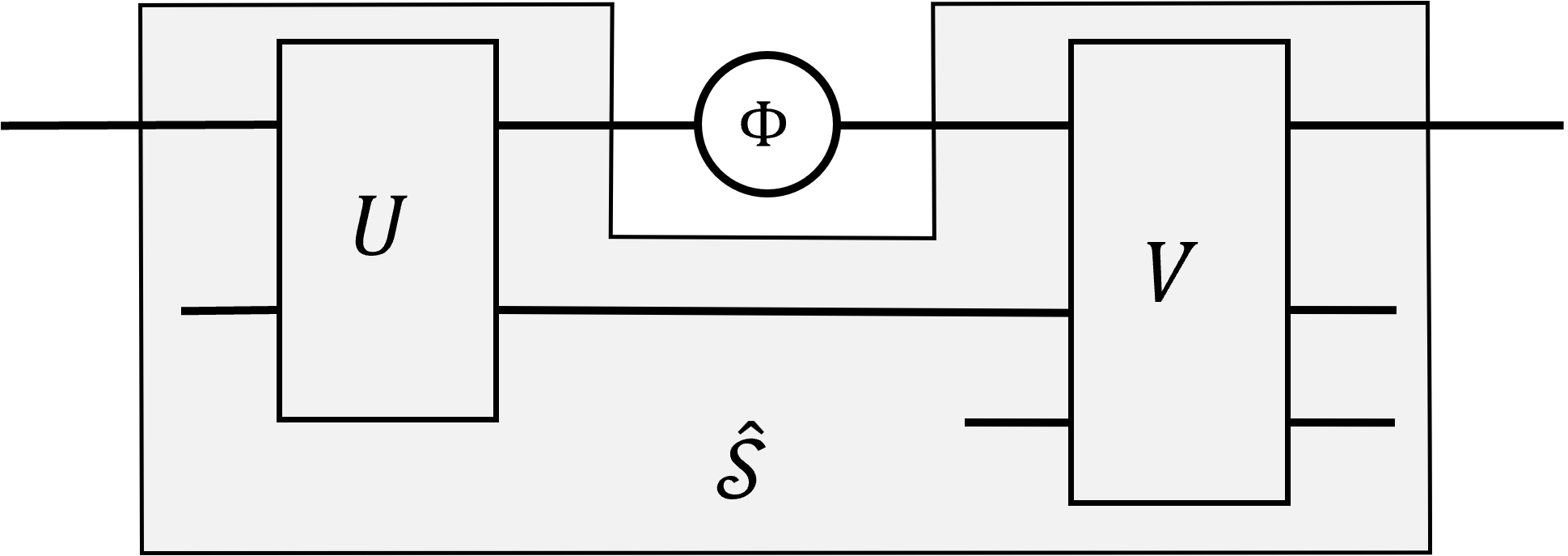}
    \caption{A schematic of quantum superchannel $\hat{\C S}$ (grey area) containing a pre $U$ and a post $V$
    unitary operation (boxes) acting on a channel $\Phi$ (circle).}
    \label{fig:sc}
\end{figure}

The operations on Choi states are further known as 
quantum superchannels~\cite{CDP08a}.
For notation, we use a hat on the symbols for superchannels.
The circuit representation of a superchannel is 
\be \hat{\C S} (\Phi)(\rho)= \text{tr}_a V \; \Phi\; U (\rho\otimes |0\ket \bra 0|),  \ee
here $U$ and $V$ are unitary operators,
and a is an ancilla. 
See Figure~\ref{fig:sc} for an illustration.
The dimension of $V$ can be larger than $U$~\cite{WW23},
but we do not need the details here. 
It is clear that higher-order superchannels can also be defined 
by an iterative use of the channel-state duality,
and this will be used for our resource theory of superchannels.

\subsection{Quantum resource theory}

We follow standard QRT~\cite{CG19} and our recent framework~\cite{W21_model,W23_ur,W23_qvn}. 
Given a set $\C D$,
a resource theory on it is defined by a tuple 
\be (\C F, \C O, \C R) \ee 
with $\C F \subset \C D$ as a set of free elements,
$\C O: \C F \ra \C F$ the set of free operations, 
and $\C R := \C D \backslash \C F$ the set of resources~\cite{CG19}.
The framework of QRT is quite broad.
For instance, the set $\C D$
can be the set of density operators acting on a Hilbert space,
or a set of unitary operators. 
A measure of resource can be defined but here we do not need to review it.

A universal resource theory is further defined as 
\be (\C F, \C O, \C R, \C U) \ee
with an additional set $\C U \subset \C R$ as 
the set of universal resource,
compared with a usual resource theory. 
The universality means that 
$ \C O (\C F \otimes \C U) $
can simulate any other process $ \C O (\C F \otimes \C R) $ efficiently. 
Here, efficiency means that the costs for the free operations $\C O$,
free elements $\C F$, and universal resources $\C U$ 
all do not grow exponentially fast with the size of the given process.
Another way to see this is that 
universal resource $\C U$ optimizes the resource measure in $\C R$.

For a given total set $\C D$,
we can define a hierarchy of resource theories.
If $\C F_1 \subset \C F_2$ for two resource theories, 
then $\C O_1 \subset \C O_2$,
and $\C R_1 \supset \C R_2$, i.e.,
more elements are treated as resourceful in the first theory.
However, to achieve universality, 
more resource power is needed for the first theory, 
denoted as $\C U_1 \succ \C U_2$.
Then there exist the following conversions between two universal resources 
\be (\C O_2 \backslash \C O_1 ) u_2 =u_1, \quad  \C O_1 u_1 =u_2, \label{eq:rconv}\ee 
for universal resource $u_{1,2}\in \C U_{1,2}$, 
modulo free elements.

\subsection{Entropy and distance measures}

The von Neumann entropy is defined as 
\be H(\rho)=\log d-R(\rho \| \pi_d), \ee 
for the quantum relative entropy $R(\rho \| \sigma)=\text{tr} \rho \log \rho - \text{tr} \rho \log \sigma$, $\forall \rho, \sigma \in \C D(\C X)$ 
(see~\cite{Wat18} for more details).
It is monotonic $H(\Phi(\rho))\geq H(\rho)$ under unital channels $\Phi$
but not for non-unital ones.
Actually this is a simple example of a QRT,
with $\pi_d$ as the free set, 
unital channels as free operations,
and $R(\rho \| \pi_d)$ as a measure of resource for all non-identity states $\rho$. 

Given a channel $\Phi$,
the coherent information $I_c(\rho,\Phi)$ is defined as~\cite{BNS98}
\be I_c(\rho,\Phi):=H(\Phi(\rho))-H(\Phi\otimes \I_d (|\varphi_\rho\ket) ), \ee
for $|\varphi_\rho\ket$ as a purification of $\rho$.
The quantity $$J(\rho,\Phi):=H(\rho)+ I_c(\rho,\Phi)$$ is the quantum mutual information,
which is always nonnegative.
Let \be I(\Phi):=I_c(\pi_d,\Phi), \ee
which, plus $\log d$, is the quantum mutual information contained in the Choi state $\omega_\Phi$.
It is clear it is additive.
Also $J(\Phi):=\max_\rho J(\rho,\Phi)$ is additive~\cite{AC97},
but $I_c(\Phi):=\max_\rho I_c(\rho,\Phi)$ is not.
These quantities are used to define quantum capacities.

To quantify the distance between channels, 
we use the fidelity between Choi states 
\be F_E(\Phi, \Psi):=  F(\Phi \otimes \I (\omega),  \Psi \otimes \I (\omega)), \ee
for the fidelity $F(\rho,\sigma):=\|\sqrt{\rho}\sqrt{\sigma}\|_1^2$,
with $\|\cdot\|_1$ denoting the trace norm.

\begin{figure}[t!]
    \centering
    \includegraphics[width=0.4\textwidth]{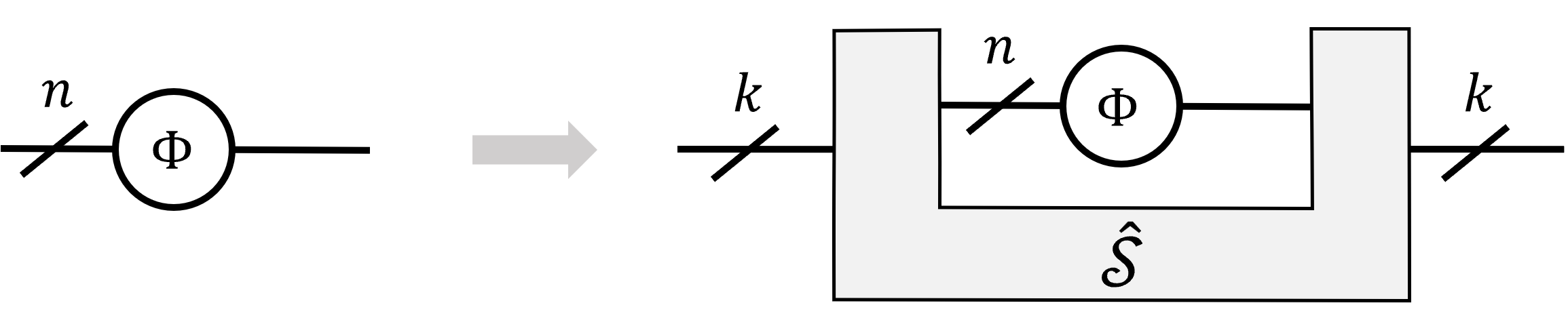}
    \caption{A schematic of quantum coding with a 
    superchannel $\hat{\C S}$ serving as the coding operation converting 
    $n$ uses of a channel $\Phi$ into $k$ uses
    which approximates the identity channel.}
    \label{fig:code}
\end{figure}

\section{Coding models}\label{sec:coding}

A coding task refers to the conversion of $n$ uses of a noise channel $\Phi$
into $k$ approximate uses of the identity channel.
The noise channel may depend on some parameters, $\mu$,
but we will simply denote it as $\Phi$. 
The value $r:={k}/{n}$ is called the coding rate for 
the encoding of $k$ qubits into $n\geq k$ qubits.
The operations on channels are in general superchannels, see Figure~\ref{fig:code}.

\begin{definition}
[Quantum coding] A coding scheme for a channel $\Phi$ is a superchannel $\hat{\C S}$ so that 
\be F_E(\I^{\otimes k},\hat{\C S} (\Phi^{\otimes n}))\geq 1-\epsilon, \label{eq:codeerror}\ee 
with $\epsilon\in [0,1]$ and positive integers $n$ and $k$. 
\end{definition}
Note a primary requirement on the coding is that the tight error bound $\epsilon$
should be smaller than $1-F_E(\Phi^{\otimes k}, \I^{\otimes k})$,
which we name as the ``bare error'' of the channel. 
The coding scheme $\hat{\C S}$ may also contain tunable parameters, $\lambda$.
A code is approximate in general since the coding error $\epsilon$ is not zero,
but it is called quasi-exact~\cite{WZO+20} if $\epsilon(\mu, \lambda, k, n) \ra 0$ occurs
in the parameter space of $(\mu, \lambda, k, n)$.
It becomes exact when $\epsilon$ is exactly zero, 
corresponding to the exact error-correction condition~\cite{KL97}.

An important fact is that realizing superchannels requires ancillary system which is assumed to be noise-free~\cite{WW23}.
Therefore, to justify a coding scheme 
it should fit into the physical settings properly. 
This actually leads to different types of channel capacities 
and the use of QRT to characterize them. 
Here we define four coding models via two notions of locality.

\begin{definition}[Transmission locality]
A quantum communication task from Alice to Bob is transmissionally local 
if there is no pre-shared entanglement between them. 
\end{definition}

When Alice and Bob both are multipartite
and there is a separable partition of their subsystems,
we can define a logical locality.

\begin{definition}[Logical locality]
A quantum communication task from Alice to Bob is logically local 
if there is no entanglement among the subsystems of Alice, 
and the subsystems of Bob. \label{def:logicall}
\end{definition}

\begin{figure*}
    \centering
    \includegraphics[width=0.7\textwidth]{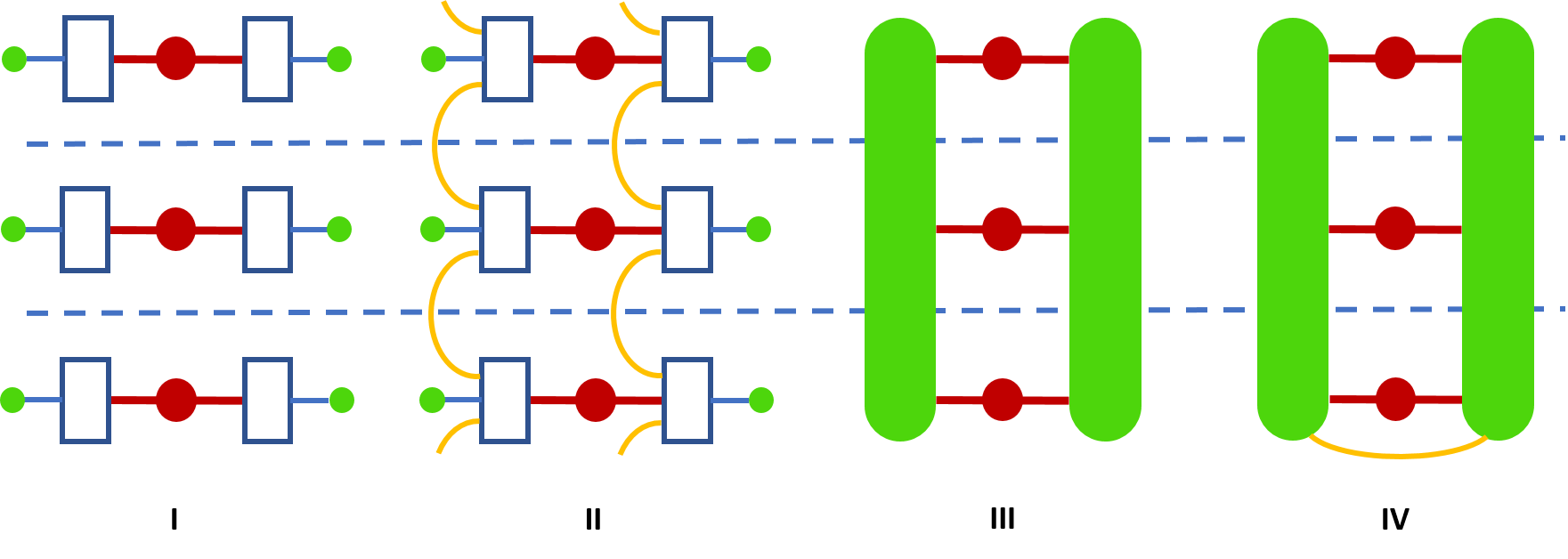}
    \caption{Four models of quantum coding.
    (a) Model I: spatially and logically local.
    The quantum capacity is given by $I(\Phi)$.
    (b) Model II: bipartite correlation between code blocks are allowed, 
    and both the encoding and decoding are logically semi-local. 
    The correlation between code blocks within a player can simulate spatial correlation between corresponding code blocks for the two players. 
    Here the correlation can be simply understood as ebits. 
    The quantum capacity is given by $(I(\Phi)+\log d)/2$.
    (c) Model III is the standard setting for quantum capacity.
    It is logically nonlocal.
    (d) Model IV is the standard setting for entanglement-assisted quantum capacity.
    It is both spatially and logically nonlocal.
    Notation: red balls are a few parallel use of noise channels, blue boxes are local coding, the giant green boxes 
    are nonlocal coding, yellow curves are entanglement assistance.}
    \label{fig:QEC_model}
\end{figure*}

These localities can be well understood by treating them
as variations of the locality or separability to define entanglement~\cite{HHH+09}.
We can then introduce in four coding models
shown in Figure~\ref{fig:QEC_model} forming a coding family.
Roughly speaking, the model I is both spatially and logically local, 
model II is logically, or equivalently spatially, semi-local, 
model III is spatially local but logically nonlocal, 
model IV is both spatially and logically nonlocal.
There can also be other models,
and in principle, there could be an indefinite number of models in the family.
For instance, there are settings when model III is assisted by 
one-way or two-way classical communication~\cite{BDS97}.

Models I and II are defined in our recent paper~\cite{W22_ca},
and was called the ``refined'' setting~\cite{W22_ca},
but here we call it ``local'' for model I, and ``semi-local'' for model II.
They can be treated as the local versions of model III and model IV, respectively.
The model III is the standard setting for quantum capacity~\cite{Dev05},
and model IV is the standard entanglement-assisted setting~\cite{BSS+99}. 

\section{Quantum capacities}\label{sec:capacity}

In this section, we study quantum capacities in these four models.
A quantum capacity can in general be defined as follows.

\begin{definition}
(Quantum capacity of a channel) 
Let $\Phi\in C(\C X, \C Y)$ be a channel 
and $k=\lfloor \alpha n \rfloor$ 
for all but finitely many positive integers $n$ and an achievable rate $\alpha \geq 0$, 
there exists coding superchannel $\hat{\C S}$ defined in a coding model such that 
\be F_E (\I^{\otimes k}, \hat{\C S} (\Phi^{\otimes n}) ) \geq 1-\epsilon\ee
for every choice of a positive real number $\epsilon$
and the quantum capacity of $\Phi$, denoted $Q(\Phi)$,
is defined as the supremum of all $\alpha$.
\end{definition}

\subsection{Models I and III}

For the model I and model III, the coding splits into 
a pair of encoding $\C E$ and decoding $\C D$ operations. 
The quantum capacity in model III is the standard notion of quantum capacity,
and has been proven to equal to the entanglement generation capacity~\cite{Dev05}. 
Model I is a special case of model III, 
and its capacity has been proven to equal to an ebit-distribution capacity~\cite{W22_ca}. 
A feature of model I is that the encoding preserves identity 
$\C E (\I^{\otimes k})=\I^{\otimes n}$.
This can be realized by a mixture of encoding isometry but with random ancillary state,
namely, $\text{tr}_a U(\omega^{\otimes (n-k)} \otimes \I^{\otimes k})$
with $(n-k)$ pair of ebits, half of which are acted upon by $U$ 
while the other half are traced out.
As has been shown~\cite{W22_ca}, this leads to a single-letter capacity for model I.
Below we  present the proofs of the two quantum capacities 
following a unified method. 
We will follow the book~\cite{Wat18} and Theorem~\ref{the:dev} below is  Theorem 8.55 in it.

\begin{theorem}
(Quantum capacity theorem: model I~\cite{W22_ca}) 
The quantum capacity of a channel $\Phi$ defined in model I is 
$Q_\textsc{i}(\Phi)=I(\Phi)$.
\end{theorem}
\begin{theorem}
(Quantum capacity theorem: model III~\cite{Wat18}) 
The quantum capacity of a channel $\Phi$ defined in model III is 
$Q_\textsc{iii}(\Phi)=\lim_{n\ra \infty}\frac{I_c(\Phi^{\otimes n})}{n}$.
\label{the:dev}
\end{theorem}

\begin{proof}
Theorem 8.53 in Ref.~\cite{Wat18} proves $I(\Phi)\leq Q_\textsc{i}(\Phi)$,
together with Theorem 8.54 in Ref.~\cite{Wat18} gives $I_c(\Phi^{\otimes n})\leq n Q(\Phi)$.
To prove $I_c(\Phi^{\otimes n})\geq n Q_\textsc{iii}(\Phi)$,
 Theorem 8.55 in Ref.~\cite{Wat18} refers to the entanglement generation scheme,
namely, for any rate $\alpha \leq Q_\textsc{iii}(\Phi)$,
there exists a state $|u\ket \in \C X^{\otimes n} \otimes \C Z^{\otimes k}$ and a decoding channel 
$\C D \in C(\C Y^{\otimes n}, \C Z^{\otimes k})$ such that 
\be F(\omega^{\otimes k}, (\C D \Phi^{\otimes n} \otimes \I^{\otimes k}) (|u\ket))\geq 1- \epsilon.\ee
Then the result follows from  
\be H(\C D \Phi^{\otimes n} \otimes \I^{\otimes k} (|u\ket))\leq 2\delta m+1, \ee
and 
\be H(\C D \Phi^{\otimes n} (\rho))\geq m-\delta m-1, \ee
and the data-processing inequality~\cite{BNS98},
for $\rho$ as the reduced state of $|u\ket$ on $\C X^{\otimes n}$.
To prove $I(\Phi)\geq Q_\textsc{i}(\Phi)$,
it specifies to ebit-distribution with $|u\ket$ replaced by $\omega^{\otimes k}$
and encoding with $\C E (\I^{\otimes k})=\I^{\otimes n}$.
The above two inequalities become
\be H(\C D \Phi^{\otimes n} \otimes \I^{\otimes k} (\omega^{\otimes k}))\leq 2\delta m+1, \ee
and 
\be H(\C D \Phi^{\otimes n} (\pi^{\otimes k}))\geq m-\delta m-1, \ee
which implies $I(\Phi)\geq Q_\textsc{i}(\Phi)$.
\end{proof}

By comparing the above two theorems, 
we see in model III
it allows more general $|u\ket$ and its reduced state $\rho$,
instead of copies of ebit $\omega$ and the completely mixed state $\pi$.
This reflects the difference between ebit distribution and entanglement generation.
For the former, only product of ebits are allowed, 
corresponding to local coding schemes,
while for the later the state $|u\ket$ can be multipartite entangled,
corresponding to general nonlocal coding schemes.

Also the preservation of identity is important for the proof of model-I capacity,
which was implicitly used~\cite{W22_ca} but not emphasized.
This relies on random encoding which is not isometric but 
isometric encoding will suffice~\cite{BKN00}.
If a mixture of isometric encodings guarantees a high entanglement fidelity $F_E$,
then each of the isometric encoding also works.
This also applies to model III,
but for models II and IV below, 
the encoding is not isometric if ignoring the noiseless entanglement assistance.

\begin{table*}[htb]
    \centering
    \caption{A table of the channel capacities studied and mentioned in this work.}
    \begin{tabular}{|c|c|c|c|c|}\hline 
                & Model I & Model II & Model III & Model IV \\ \hline 
        Quantum &  $Q_\textsc{i}=I(\Phi)$  & $Q_\textsc{ii}=[\log d+I(\Phi)]/2$ & 
        $Q_\textsc{iii}=\lim_{n\ra \infty}\frac{I_c(\Phi^{\otimes n})}{n}$ & $Q_\textsc{iv}=J(\Phi)/2$ \\ \hline 
        Classical & $C_\textsc{i}=\chi_\textsc{ort}(\Phi)$ & 
        $C_\textsc{ii}=\log d+I(\Phi)$ & 
        $C_\textsc{iii}=\lim_{n\ra \infty}\frac{\chi(\Phi^{\otimes n})}{n}$ & 
        $C_\textsc{iv}=J(\Phi)$ \\ \hline 
        Private & $P_\textsc{i}=\Delta C_\textsc{i}$ & 
        $P_\textsc{ii}=[\log d+I(\Phi)]/2$ & 
        $P_\textsc{iii}=\Delta C_\textsc{iii}$ & $P_\textsc{iv}=J(\Phi)/2$  \\ \hline 
    \end{tabular}
    \label{tab:capacity}
\end{table*}

\subsection{Models II and IV}

Models II and IV are entanglement-assisted (EA).
In model II,
due to the logical locality only bipartite entanglement is allowed,
and the perfect resource is ebit. 
For model IV, any pre-shared entangled state is allowed.
Note that in model II, the shared ebits are within each player, Alice or Bob.
These pre-shared state can be understood as being generated by a pre round of 
coding on it, and then we do not need to consider noises on it anymore.
For notation, the model II was 
simply denoted as ``EA'' in our previous paper~\cite{W22_ca}.
A simple but important fact to verify is that 
this can be used to generate remote ebits between the two players by quantum teleportation.
Also remote ebits can be used to generate ebits at Alice's or Bob's side.
Therefore, the logical semi-locality is equivalent to spatial semi-locality.

In the EA settings, 
an important phenomenon is that quantum capacity is half of its classical capacity
of a channel based on quantum teleportation and dense coding~\cite{Wat18}. 
A technical part is the usage of EA Holevo quantity $\chi_\textsc{ea}$ for the standard EA setting~\cite{Hol99},
which is restricted to an orthogonal case,
denoted `EAO', due to the usage of ebit-assistance~\cite{W22_ca}.
The classical capacities are firstly expressed as Holevo quantities 
$$C_\textsc{iv}(\Phi)=\lim_{n\ra \infty} \frac{\chi_\textsc{ea}(\Phi^{\otimes n})}{n}$$
and $C_\textsc{ii}(\Phi)=\chi_\textsc{eao}(\Phi)$,
and then related to coherent information.
For simplicity, we recall the two theorems below and reproduce a unified proof.
Theorem~\ref{the:bss} below is  Theorem 8.41 in Ref.~\cite{Wat18}.

\begin{theorem}
(Quantum capacity theorem: model II~\cite{W22_ca}) 
The quantum capacity of a channel $\Phi$ defined in model II is 
$Q_\textsc{ii}(\Phi)=(\log d+ I(\Phi))/2$.
\end{theorem}

\begin{theorem}
(Quantum capacity theorem: model IV~\cite{BSS+99}) 
The quantum capacity of a channel $\Phi$ defined in model IV is 
$Q_\textsc{iv}(\Phi)=J(\Phi)/2$.
\label{the:bss}
\end{theorem}

\begin{proof}
In Lemma 8.39 of Ref.~\cite{Wat18}, it proves 
\be \chi_\textsc{ea}(\Phi)\geq H(\pi_d)+I_c(\pi_d, \Phi) \ee
for $\pi_d=\Pi/d$ and $d=\text{tr}(\Pi)$ and $\Pi$ is a projector.
The proof actually proves the stronger result
\be \chi_\textsc{eao}(\Phi)\geq H(\pi_d)+I_c(\pi_d, \Phi) \ee
as it uses a completely uniform ensemble of Bell states, $\eta_*$, which is an EAO ensemble.
Applying Lemma 8.36 in Ref.~\cite{Wat18} leads to $C_\textsc{ii}(\Phi) \geq \log d +I(\Phi)$
and $C_\textsc{iv}(\Phi)\geq J(\Phi)$.

 Lemma 8.40~\cite{Wat18} applies to EA ensemble $\eta$ gives 
\be \chi(\Phi\otimes \I (\eta))\leq H(\sigma)+I_c(\sigma,\Phi). \ee
It also applies to any EAO ensemble $\eta_\textsc{eao}$ which becomes 
\be \chi_\textsc{ort}(\Phi\otimes \I (\eta_\textsc{eao}))\leq \log d +I(\Phi). \ee
This leads to $C_\textsc{ii}(\Phi) \leq \log d +I(\Phi)$
and $C_\textsc{iv}(\Phi)\leq J(\Phi)$.
This completes the proof.
\end{proof}

Therefore, we established the quantum capacities for the four models above,
with only the standard model III having a non-additive measure of capacity.
Due to the additivity of capacities for models I, II, and IV, 
the converse quantum Shannon theorem in these models are easy to prove by
following the well established methods~\cite{Win99,ON99,BDH+14,TWW17}.
The capacity in each of the three models also serves 
as the strong converse capacity, meaning that once a rate is larger 
than a capacity, the coding error would converge exponentially fast to 1 
in the asymptotic limit for all possible codings. 

Our framework also applies to classical capacities and private capacities~\cite{Wil17}. 
Here we would not reproduce the details~\cite{W22_ca}.
The channel capacities studied in this work are summarized in Table~\ref{tab:capacity}.
The quantity $\chi$ is the Holevo quantity~\cite{Hol99} 
and $\chi_\textsc{ort}$ here is the Holevo quantity when the classical-to-quantum 
encoding is restricted to being isometric~\cite{W22_ca}.
The private capacities are equal to the quantum ones
for the entanglement-assisted cases, 
and otherwise, they are from the Holevo quantity of 
a channel minus its complementary channel,
hence the notation $\Delta$ in the table.

\section{QRT of coding}\label{sec:qrt}

We now use quantum resource theory (QRT) to characterize the coding models.
As codings are superchannels, 
the QRT of codings is a QRT of superchannels. 
Due to the channel-state duality~\cite{Cho75}, 
it is not hard to formulate it by referring to QRT of other kinds,
especially the QRT of channels~\cite{CG19}. 
Here, the set of objects we consider are superchannels 
that are used for codings. 

To define a family, we need to make sure 
the goal of each model is the same.
For coding, the goal is indeed the same, 
which is to convert a channel into the perfect identity channel with high accuracy. 
The universal resources are the codings that 
achieve the capacity of a channel.

\begin{definition}[QRT of codings]
A QRT of codings $(\C F, \C O, \C R, \C U)$ for a channel $\Phi$ is defined by
a proper set of free superchannels $\C F$ used in a coding model,
which are transmissionally local 
and can only preserve or increase its bare error,
a free set $\C O: \C F \ra \C F $,
and the resource $\C R$ is formed by all allowed
superchannels that can decrease the bare error. 
The universal set $\C U \subset \C R$ contains the codings that 
achieve the capacity of a channel $\Phi$ in a coding model.
\end{definition}

Furthermore, a hierarchy can be defined based on a subset structure of free sets.
We say such coding models belong to a coding family.
The model I have the largest $\C F$, while model IV the smallest, 
but model I has the least powerful $\C U$, while model IV has the most powerful $\C U$.
To prepare for the theorem below,
we clarify a few points. 
All one-way classical communication from Alice to Bob is free in all the models.
For model II, we say an operation is semi-local when 
a local operation can be assisted by bipartite entanglement.
For model III, the back classical communication from Bob to Alice is not allowed,
however, since it will change its capacity~\cite{BDS97}.
Although model III is logically nonlocal but it is spatially local,
so it is not comparable with model II,
and this leads to two sub-hierarchies.

\begin{theorem}[Hierarchy of coding models]
The models I, III, and IV form a hierarchy,
and models I, II, and IV also form a hierarchy in the coding family.
\end{theorem}
\begin{proof}
We prove the theorem by the explicit construction of QRT for each coding model. 
Relative to a logical locality and the transmission locality, 
model I is defined by the free set $\C F_\textsc{i}$ which  
can only preserve or increase the bare error.
For instance, all one-side processing at Alice's or Bob's side is free.
The set $\C F_\textsc{ii} \subset \C F_\textsc{i}$ is 
logically biseparable or semi-local which preserves or increases the bare error.
This selects some ebit-assisted semi-local codings as resources. 
The set $\C F_\textsc{iii} \subset \C F_\textsc{i}$ is transmissionally local but without back classical communication, and logically being $l$-local for $l\in o(n)$ as a small value compared with $n$.
Finally, $\C F_\textsc{iv} \subset \C F_\textsc{ii}, \C F_\textsc{iii}$ 
only allows transmissionally local and logically product operations.
It is also clear there is no set relation between $\C F_\textsc{ii}$ and $\C F_\textsc{iii}$.

Meanwhile, $\C R_\textsc{i} \subset \C R_\textsc{ii} \subset \C R_\textsc{iv}$,
and $\C R_\textsc{i} \subset \C R_\textsc{iii} \subset \C R_\textsc{iv}$.
The subset structures can be shown case by case.
For instance, 
while $\hat{\C S}\in \C R_1$ are the case of $Q_\textsc{i}(\Phi)\geq 0$,
the case of $Q_\textsc{i}(\Phi)=0$ leads to the existance of
multipartite entangled  codings $\hat{\C S}\in \C R_3 \cap \C F_1$
or ebit-assisted codings $\hat{\C S}\in \C R_2 \cap \C F_1$.
Similar arguments also apply to other cases.

To verify the universal resource conversion~(\ref{eq:rconv}),
it is enough to observe that
a coding that works at a higher level also works at a lower level, 
and can be reduced to one that works at a lower level. 
For instance, by ignoring the assisted entangled state $|\eta\ket$ in model IV it reduces to model III.
On the contrary, $|\eta\ket$ can be prepared in model III 
and further used as the assistance in model IV. 
\end{proof}

\subsection{Applications}\label{sec:app}

The benefit of using QRT is to understand these models systematically,
with the capacities as the measure of universal coding resources.
It also highlights the role of local computational ability,
and the trade-off between local computation and codings. 
The local computation belongs to their free sets. 
These models can be chosen for different practical situations.
Suppose a communication task is to send a large amount of data,
namely, highly entangled states,
over a noisy channel.
Below we analyze different strategies in these models.

\begin{itemize}
    \item For model I, Alice and Bob have the largest pre/post computational ability.
Alice can represent the data $|\psi\ket$ as a quantum circuit, $U$,
with \be |\psi\ket=U|\psi_1,\psi_2,\cdots, \psi_k\ket. \label{eq:blockst}\ee
Alice can use a classical channel to send the classical representation of the circuit, $[U]$,
namely, the type and location of each elementary gate in it,
and then only send unentangled qubits over the quantum channel to Bob.
Bob has to perform $U$ according to $[U]$ in order to obtain $|\psi\ket$.
Note $U$ might be of high depths. \\
This scheme can be applied in many settings. 
For instance, for channels that are only slightly noisy
so that no powerful encoding is needed, 
or in blind quantum computation~\cite{BFK09} 
when Alice wants to hide the input data from Bob
but not the algorithm itself,
or in teleportation-based models 
such as the cluster-state model~\cite{RB01} and quantum von Neumann 
architecture~\cite{W22_qvn} wherein a computation is simulated 
by a sequence of gate teleportation on initially unentangled states.
  \item  If the channel is quite noisy, 
one can move on to the model II which always has a nonnegative quantum capacity. 
The ebits play essential roles here.
We find, interestingly, Alice can represent the data $|\psi\ket$ as a
matrix-product state (MPS)~\cite{PVW+07} 
\be |\psi\ket= \sum_{i_1,\dots,i_N} \text{tr} (B A^{i_N} \cdots A^{i_1} ) |i_1\dots i_N\rangle,
\label{eq:mps}\ee
with local tensors $A^{i_n}$ (and a boundary operator $B$),
and then send its circuit representation $[A^{i_n}]$ to Bob,
who can then use them and also ebits as resources to obtain $|\psi\ket$
by only applying constant-depth local operations.
As is well known, MPS are proper forms 
to characterize entanglement, and play essential roles in many-body physics
to describe topological order~\cite{ZCZ+15}.
Here a MPS can be shared remotely by a few parties, 
and may have applications in distributed computing~\cite{WEH18}.
  \item The models III and IV are well known.
From our perspective,
the model III allows any nonlocal encodings 
and any pre and post $l$-local operations on 
the state $|\psi\ket$.
This can indeed describe some nontrivial operations on codes,
such as code switching~\cite{PR13,W21_model},
which is an important scheme to realize universal set of logical gates.
For model IV, Alice would send $|\psi\ket$ as a whole to Bob,
so the only required ability for Bob is to store 
and manipulate each qubit.
\end{itemize}

If for a channel $I(\Phi)\geq 0$, then all the models would work,
with higher-level models achieving larger capacities.
When $I(\Phi)\leq 0$, one has to move on,
e.g., with remote ebits that may be generated by model I 
for another channel $\Psi$ with $I(\Psi)\geq 0$,
or choose other models.
Also note that the models I and II are not the `one-shot' versions, 
since the one-shot setting only allows separable codings,
while I and II allow entangling isometry as codings.

\subsection{Classification of codes}\label{sec:class}

Our classification theory is not only useful for choosing a proper coding model 
in practice, but also for the usage of error-correction codes. 
In the setting of noncooperative communication~\cite{ML14},
the classification and recognition of codes are important. 
When Alice and Bob do not mutually agreed upon the code being used,
Bob has to recognize the type of code in order to
choose a proper decoding algorithm,
which should be adapted to the right type of code.

We know from the classical coding theory~\cite{RL09}
a code is in general being block or convoluntional.
The later is characterized by a temporal order of data
and memory between the encoding of blocks of data,
while for block codes, 
each block of data encodes and decodes separately.
For quantum codes, it turns out the memory effect can be simulated by ebits 
via teleportation, inducing a temporal order of data blocks.
It is then easy to see convoluntional codes 
can be described by model II (using teleportation module among 
the encoding parts),
and also model IV (when many of its large blocks are available).
But note the EA settings do not have to be convoluntional, 
say, when only one block with EA is used for coding. 
The difference between models II and IV, 
and also models I and III, 
is a block is being `small' or `large'. 
This can be characterized by the depth of the encoding circuit for a code,
see Table~\ref{tab:code-type}.

\begin{table}[t!]
    \centering
    \caption{A table of the code types induced by the four coding models 
    classified in this work
    and some stabilizer code examples in literature.
    Here the depth refers to the depth of the encoding circuit for a code.
    Note more refined classification is possible
    by considering more features of codes.}
    \begin{tabular}{|c|c|c|c|c|}\hline 
                & Model I & Model II    & Model III & Model IV \\ \hline 
        Type &  small block     & convolutional & large block & convolutional \\ \hline 
        Depth & small & small & large & large \\ \hline 
        Examples & \cite{Shor96,NC00,WWC+22} & \cite{Chau98,OT03,Wil09,FP14,PTO09} 
        & \cite{Kit03,BE21,BKX+20,Vid07} & \cite{RDR12,WG13,BDH06} \\ \hline 
    \end{tabular}
    \label{tab:code-type}
\end{table}

A large depth of the encoding circuit with local gates 
can lead to large values of entanglement increasing with the system size $n$,
while a small or constant depth circuit cannot do so.
There are many ways to characterize entanglement~\cite{PVW+07,HHH+09,ZCZ+15},
one powerful approach is to express states as MPS
(\ref{eq:mps}),
and the bipartite entanglement entropy $S_E$ in a state satisfies
\be S_E \propto \log \chi \ee 
for $\chi$ as the bond dimension of the entanglement space, 
which is used to carry the logical information.

Small block codes and convolutional codes
will have small values of entanglement entropy $S_E$.
Examples of small block codes are those 
with small $k$ and small encoding circuit depths,
including some well-known small-size codes~\cite{Shor96,NC00}
and symmetry-protected codes~\cite{WWC+22},
and also some convolutional codes~\cite{Chau98,OT03,Wil09,FP14},
and quantum Turbo codes, 
as interleaved convoluntional codes~\cite{PTO09}.

Meanwhile, large codes are those with a large value of $k$,
and most likely also the rate $r=k/n$. 
Notable examples are some high-rate LDPC codes~\cite{BE21}, 
Polar codes~\cite{BKX+20}, the MERA codes~\cite{Vid07} etc.
Some of them are entanglement-assisted~\cite{RDR12,WG13}, 
hence can be expanded to convolutional ones 
when many blocks are used. 

We can also consider the \emph{reverse} problem:
given a code that is promised to be one of the types we know, 
how to determine its type? 
Such a quantum code recognition task is an example of 
quantum hypothesis testing~\cite{Wat18},
but here the goal is to test its type not its formula.
In general, quantum machine learning algorithm~\cite{Wit14} can be employed to serve as 
a classifier, which, however, is resource-intensive. 
This task is also hard for the classical case, 
see Ref.~\cite{LCL24} for a latest study. 
Different from the classical case, 
a quantum test is a POVM that will consume many samples of the state $|\psi\ket$.
Here we lay out an entanglement-based scheme
but more advanced method is necessary. 

There are ways to measure $S_E$ in experiments~\cite{IMP+15,BEJ+19}. 
Once the whole state $|\psi\ket$ is obtained by Bob sent from Alice,
Bob can do a few binary test.
First, it is easy to see if it is EA or not 
since the EA side channel is noise-free and has to be established beforehand.
Then the value of $S_E$ can tell small codes from large ones.
Due to the state form~(\ref{eq:blockst}) for block codes,
    it is with high probability that far apart sites have no entanglement,
    but not the case for convolutional codes.
    We can use Bell test~\cite{Mer93} to distinguish them.
But it is not easy to distinguish convolutional from Turbo codes,
and LDPC from Polar codes by entanglement entropy. 
Therefore, 
more quantities are needed with could be other entanglement measures 
or machine-learned features that deserve further study.

\section{Numerical simulation}\label{sec:num}


To further understand quantum capacities and the nonadditivity,
here we numerically explore the gaps among the quantum capacities for 
the case of qubit channels.
A general qubit channel contains 12 parameters.
This can be seen in the so-called affine representation $\C T$
of the form 
\be \C T= \begin{pmatrix} 1 & 0 \\ \vec{t} & T \end{pmatrix},
\ee 
which is a $4 \times 4$ real matrix.
The vector $\vec{t}$ is the shift of the center of the Bloch ball,
and the matrix $T$ enables the distortion of the ball.
In order to represent a channel succinctly, 
we use $|t|$ and the Frobenius norm $\|T\|_F$ to represent a channel.
A larger $|t|$ means larger non-unitality, 
while a smaller $\|T\|_F$ means larger distortion of the ball. 
In our simulation we do not observe a clear 
dependence on $\|T\|_F$, 
so we focus on the behaviors of capacities as functions of $|t|$.
In our algorithm, given a random qubit channel $\Phi$~\cite{WBOS13}
we use an optimization algorithm from Matlab to compute the capacity quantities.
The rank of $\Phi$ is an input parameter,
and for each rank we randomly sample hundreds of qubit channels. 

\begin{figure}[b!]
    \centering
    \includegraphics[width=0.52\textwidth]{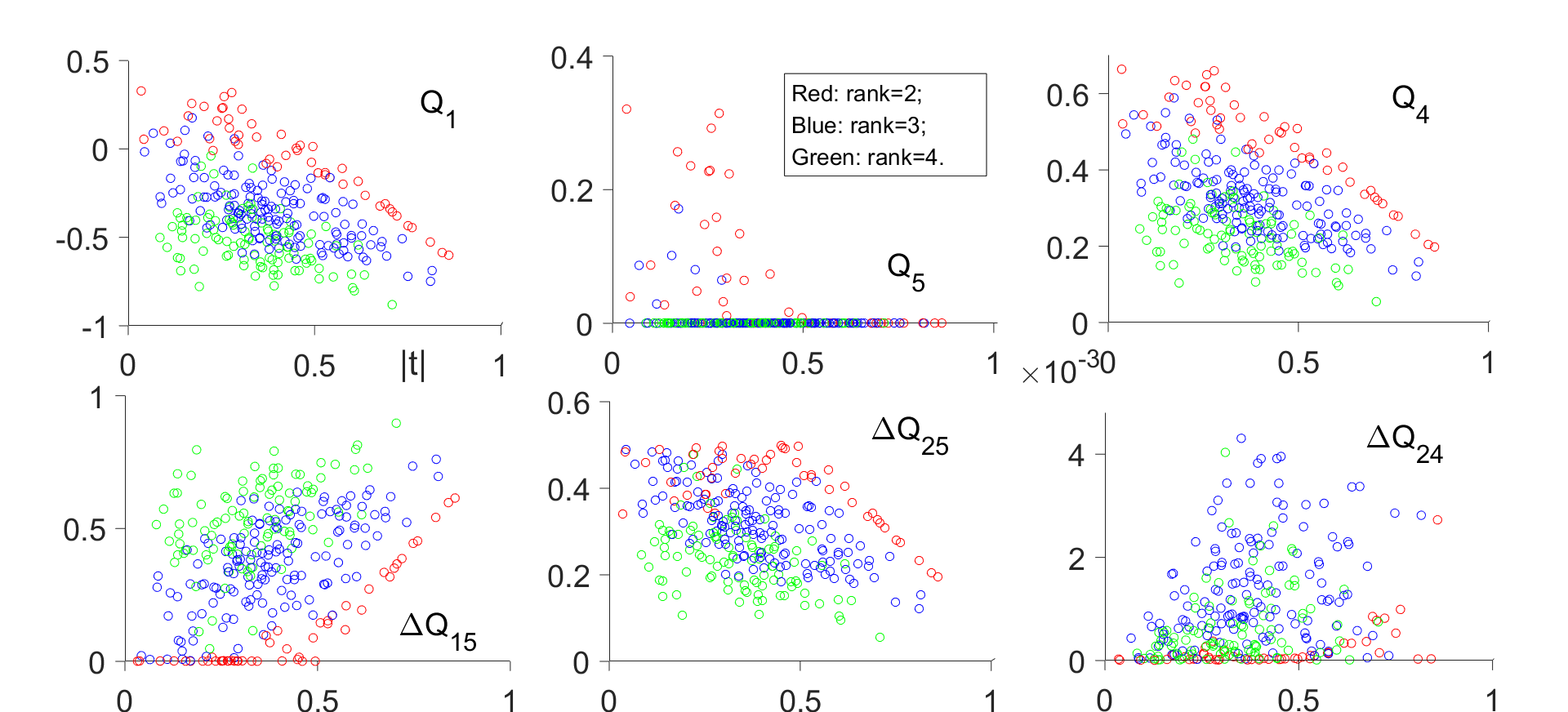}
    \caption{The quantum capacity gaps for qubit channels
    of rank two (red), rank three (blue), and rank four (green).
    Each panel is for a capacity or capacity gap.
    The horizontal axes are all  
    $|t|$, which is the size of the shift vector in the affine 
    representation of a channel.
    }
    \label{fig:capacity1}
\end{figure}

\begin{figure}[t!]
    \centering
    \includegraphics[width=0.5\textwidth]{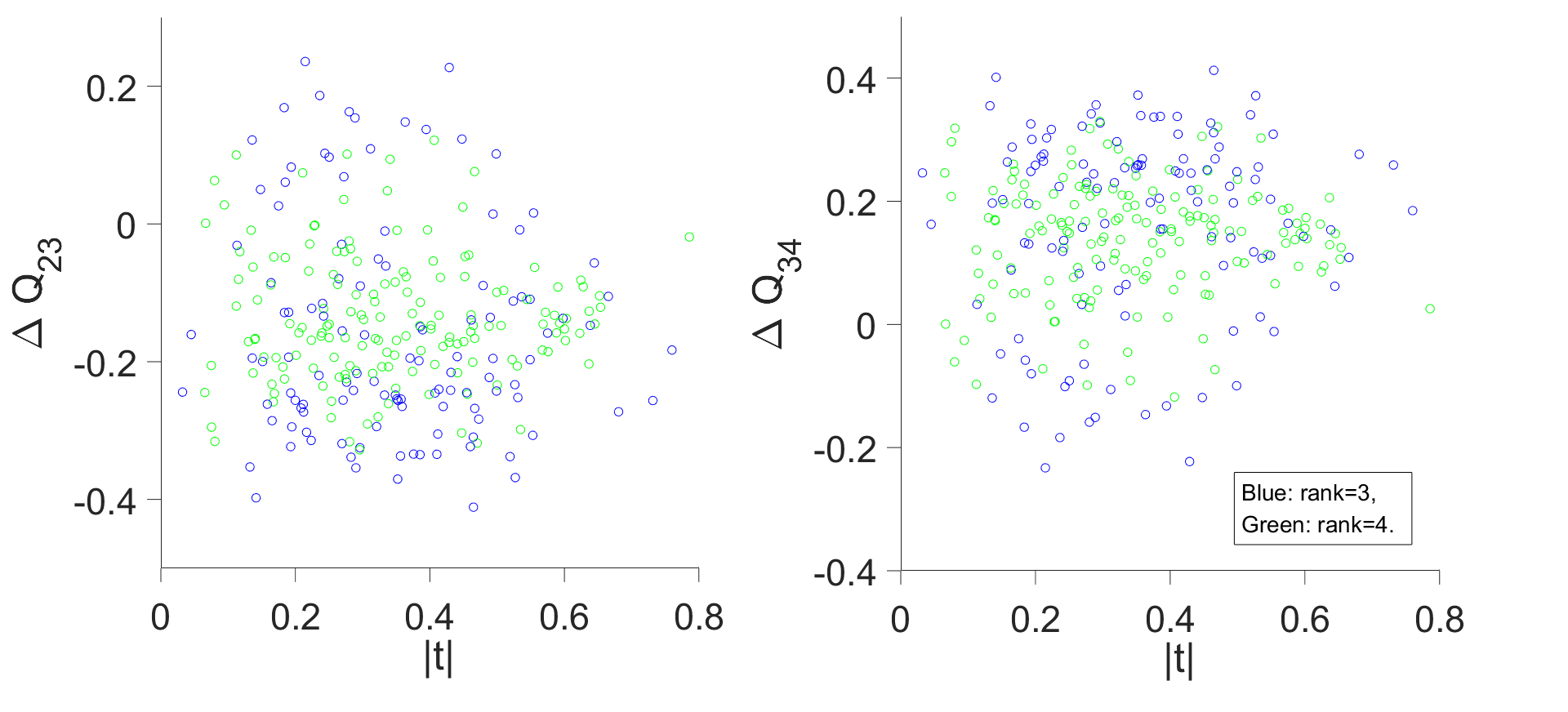}
    \caption{The quantum capacity gaps for qubit channels
    of rank three (blue) and rank four (green).
    Each panel is for a capacity gap.
    The $|t|$ is the size of the shift vector in the affine 
    representation of a channel.
    }
    \label{fig:capacity3}
\end{figure}

\subsection{Model I}

From the general relation between Holevo quantity and coherent information, 
it is easy to see 
\be Q_\textsc{i}(\Phi) \leq C_\textsc{i}(\Phi) \leq I_c(\Phi).\ee 
This means the `one-shot' quantity $Q_\textsc{v}(\Phi):=I_c(\Phi)=\max_{\sigma} I_c(\sigma,\Phi)$
serves as a good upper bound in model I, 
however, it does not serve as a quantum capacity in general. 

It is shown that
$Q_\textsc{v}(\Phi)$ is almost surely positive if $\text{r}(\Phi) \leq d$; otherwise,
it is almost surely being zero~\cite{SD22}. 
Here we numerically confirmed this for qubit channels in Figure~\ref{fig:capacity1}. 
Note we use subscript numbers to simplify the capacity quantities.
For qubit rank-two channels we find $Q_\textsc{v}(\Phi)$ are not only positive, 
but also there is a clear dependence on $|t|$.
There appears to be a transition region at about $|t| \sim 0.5$,
beyond which $Q_\textsc{v}(\Phi)$ are mostly zero,
and $Q_\textsc{i}(\Phi)$ are mostly negative and the upper envelope is almost linear with $|t|$.
For rank-three and rank-four channels, we see that $Q_\textsc{i}(\Phi)$ are mostly negative
while $Q_\textsc{v}(\Phi)$ are mostly zero.
When $Q_\textsc{v}(\Phi)=0$, the optimal input state is pure.
The gap $\Delta Q_{15}=Q_\textsc{v}-Q_\textsc{i}$,
and similarly for $\Delta Q_{25}=Q_\textsc{ii}-Q_\textsc{v}$, 
shows a clear transition region for the rank-2 case.

\subsection{Model II vs model IV}

In model IV, achieving $Q_\textsc{iv}(\Phi)$
is to achieve the regularized EA Holevo capacity,
which shall require a highly-entangled state $|\eta\ket$ 
as a resource which is not a product of ebits.
The gap 
\be \Delta Q_{24}(\Phi)= Q_\textsc{iv}(\Phi)-Q_\textsc{ii}(\Phi)\ee 
is a measure of the nonlocal assistance effect of 
an optimal resource state $|\eta\ket$. 
The $Q_\textsc{iv}$ is upper bounded by $(\log d + Q_\textsc{v})/2$,
so $\Delta Q_{24} \leq \Delta Q_{15} /2$.
We can see from the result in Figure~\ref{fig:capacity1} that 
the gap $\Delta Q_{24}(\Phi)$ is quite small (in the order $10^{-3}$) and 
is more apparent for larger shift $|t|$ for rank-2 channels.
For higher-rank cases, the gap becomes larger
and there is no obvious dependence on $|t|$.
Based on primary simulation tests not reported here,
we expect this gap will get more apparent for higher dimensional channels.

\subsection{Model III}

The quantum capacity $Q_\textsc{iii}(\Phi)$ 
cannot be explicitly computed in general,
but it can be upper bounded.
A nice method is to use the convex decomposition of channels~\cite{Smi08}.
For the qubit case, it is well known that 
a qubit channel can be decomposed as the convex sum 
\be \Phi = p \Phi^\text{g}_1+ (1-p) \Phi^\text{g}_2, \ee 
for two so-called generalized extreme channels
which are channels with rank up to two~\cite{RSW02,WBOS13}.
A qubit generalized extreme channel is either degradable or anti-degradable~\cite{WP07},
for which its quantum capacity is additive. 
Therefore, it has been proposed to use the following as an upper bound of the quantum capacity 
\be Q_\textsc{iii}(\Phi) \leq \inf (p Q_\textsc{iii}(\Phi^\text{g}_1)+ (1-p) Q_\textsc{iii}(\Phi^\text{g}_2)). \ee 
Any such decomposition would serve as a looser upper bound for its capacity,
and we denote such a value as $ Q_\textsc{iii}^\textsc{ub}(\Phi)$.
Also there is a lower bound $ Q_\textsc{iii}(\Phi) \geq Q_\textsc{v}(\Phi)$,
with equality holds for degradable channels~\cite{DS05}.
We also observe that $Q_\textsc{v}(\Phi) \leq Q_\textsc{ii}(\Phi)$.
This is explained by the behavior of $Q_\textsc{i}$:
when it is positive, $Q_\textsc{ii}\geq \frac{1}{2}$ while $Q_\textsc{v}\geq 0$, 
when it is negative, $Q_\textsc{ii}\geq 0$ while $Q_\textsc{v} =0$.

Here we plot the quantity $\Delta Q_{23}= Q_\textsc{iii}^\textsc{ub}-Q_\textsc{ii}$
and $\Delta Q_{34}= Q_\textsc{iv} - Q_\textsc{iii}^\textsc{ub}$
for random qubit channels
of rank three (blue)  and four (green) in Figure~\ref{fig:capacity3}. 
We see that both $\Delta Q_{23}$ and $\Delta Q_{34}$ can be either positive or negative.
There is no clear dependence on the rank and $|t|$ of a noise channel.
Most of $\Delta Q_{34}$ is positive, while
most of $\Delta Q_{23}$ is negative.
A negative $\Delta Q_{34}$ means the bound is not tight,
while a negative $\Delta Q_{23}$ means $Q_\textsc{iii}$ is even smaller than $Q_\textsc{ii}$
for some channels. 
This indeed confirms that the model III does not perfectly lie in between 
the model II and model IV for arbitrary channels.

\section{Conclusion}\label{sec:conc}

In this work, we establish a quantum resource theory approach 
to describe a family of coding models 
that are of importance to quantum communication and error correction.
By treating codings as superchannels,
our approach is broad to describe a few important quantum capacities
and types of codes,
and may also be used to discover new ones. 

Along the line, as we have mentioned there are 
other types of models or quantities, 
including back classical communication~\cite{BDS97},
 the simultaneous classical and quantum communication~\cite{DS05},
    reverse coherent information~\cite{GPL09}, 
 the entanglement cost of channels~\cite{BBC+13},
 and the Rains information~\cite{TWW17}.
Whether proper coding models relating to them can be defined 
and put in the hierarchy of coding family are unclear.
It is also worthy to mention the codings with 
infinite-dimensional systems~\cite{HW01,WPG+12},
which require further investigation to generalize our approach.

\vskip 0.3cm
\begin{center}
{\bf Acknowledgement}
\end{center}

This work has been funded by the National Natural
Science Foundation of China under Grants 12047503
and 12105343 (D.-S.W., Y.-D.L.), 61771377 (Y.-J.W.), the Key R\&D
Project of Shannxi Province and the Natural Science
Foundation of Guangdong Province (Y.-J.W.), and the
National Key R\&D Program of China under Grant
2020YFA0712700 (S.L.).


\end{spacing}
\bibliography{ext}{}
\bibliographystyle{apsrev4-1}

\end{document}